\documentclass[11pt]{article}

\usepackage[margin=1in]{geometry}
\usepackage{amsmath,amssymb,amsthm}
\usepackage{booktabs}
\usepackage{graphicx}
\graphicspath{{figures/}}
\usepackage{hyperref}
\usepackage{natbib}
\theoremstyle{plain}
\newtheorem{lemma}{Lemma}

\theoremstyle{remark}

\newcommand{\kstar}{k^{\ast}}
\newcommand{\knaive}{k_{\mathrm{naive}}}
\newcommand{\kstarM}{k^{\ast}_{\mathrm{M}}}
\newcommand{\knaiveM}{k_{\mathrm{naive},\mathrm{M}}}
\newcommand{\kstarJ}{k^{\ast}_{\mathrm{J}}}
\newcommand{\knaiveJ}{k_{\mathrm{naive},\mathrm{J}}}
\newcommand{\nz}{n_z}
\newcommand{\La}{\Lambda_\alpha}

\title{A closed-form sample size correction for always-valid inference with optional stopping}
\author{M\aa rten Schultzberg\\
Experimentation Platform team, Confidence, Spotify\\
\texttt{mschultzberg@spotify.com}}
\date{\today}

\begin{document}

\maketitle

\begin{abstract}
Sequential tests that allow continuous monitoring are common in A/B
experimentation. Power calculations for these tests require simulations that are
hard to scale across many metrics on an experimentation platform. Instead, a
common sizing heuristic inflates the fixed-sample size until the marginal
rejection probability at the planned endpoint reaches $1-\beta$. This
last-point rule is conservative because always-valid (AV) power is the
probability of a boundary crossing
at any time during the run, not at the endpoint alone. We give a
closed-form correction factor $\kstar(\alpha, \beta, t_0)$ expressed in
elementary functions and the bivariate normal CDF, where $t_0 = m/\nz$ is
the burn-in fraction. The closed-form approximation depends on the
boundary only through its value and slope at the planned endpoint and can
be evaluated for any smooth concave boundary. We work out three cases:
the confidence sequences of \citet{wskr2023} and \citet{maharaj2023},
and the mixture sequential probability ratio test of
\citet{johariopre2022}. Setting the total sample size to $\kstar \cdot
\nz$, where $\nz$ is the fixed-sample size for allocation ratio $r$,
hits empirical power within approximately $3$ percentage points of
target in Gaussian simulations. The correction factor depends on the
allocation ratio $r$ only through $t_0 = m/\nz(r)$.
We study sensitivity to the burn-in parameter and show that the
correction saves $8$--$20\%$ of the last-point sample budget across the
operating range.
\end{abstract}

\section{Introduction}
\label{sec:intro}

Sequential tests that control type-I error under continuous monitoring without a
pre-determined max sample size are widely deployed on experimentation platforms
\citep{wskr2023,maharaj2023,johariopre2022}. They allow experiments to stop
early when evidence is clear, but sizing them correctly requires
simulation-based power calculations that are expensive at platform scale. A
common heuristic is the last-point rule \citep{growthbookdocs}, which inflates
the fixed-sample size $\nz$
until the marginal rejection probability at the planned endpoint reaches
$1-\beta$. This heuristic is conservative because AV power includes
boundary crossings before the endpoint, not only at it. At $(\alpha, 1-\beta) =
(0.05, 0.80)$ the rule delivers empirical power between $0.86$ and $0.88$,
oversizing by seven to nine percentage points.

We give a closed-form factor $\kstar(\alpha, \beta, t_0)$ such that
$n = \kstar \cdot \nz$ approximately hits target power without simulation,
where $t_0 = m/\nz$ is the burn-in fraction. The closed-form approximation
depends on the boundary only through its value and slope at the planned
endpoint and can be evaluated for any smooth concave boundary. We work out
three cases: the WSKR confidence sequence \citep{wskr2023}, the Maharaj
confidence sequence \citep{maharaj2023}, and the \citet{johariopre2022}
mSPRT. The saving over the last-point rule ranges from $8\%$ to $20\%$
across boundaries, burn-in values, and the $(\alpha, \beta)$ grid.

\section{Setup}
\label{sec:setup}

\paragraph{Two-sample experiment.}
Two arms with control mean $\mu_C$, treatment mean $\mu_T$, true effect
$\delta = \mu_T - \mu_C$, and common variance $\sigma^2$. Observations
are i.i.d.\ within arm with finite $(2+\varepsilon)$-th moment.
Allocation ratio $r = n_C/n_T \ge 1$, total sample size $n = n_T(1+r)$,
standardised effect $\Delta = \delta/\sigma$. We write
$\delta_{\mathrm{MDE}}$ for the planned minimum detectable effect.
The hypotheses are $H_0: \delta \le 0$ versus $H_1: \delta > 0$ at
one-sided level $\alpha$ and target power $1-\beta$. The fixed-sample
$z$-test sample size is \citep{cohen1988,zhouluschallah2023}
\begin{equation}
  \nz = \frac{(1+r)^2}{r} \cdot
  \frac{(z_\alpha + z_\beta)^2 \sigma^2}{\delta_{\mathrm{MDE}}^2},
  \label{eq:nz}
\end{equation}
where $z_p = \Phi^{-1}(1-p)$. Under equal allocation,
$\nz \cdot \Delta_{\mathrm{MDE}}^2 / 4 = (z_\alpha + z_\beta)^2$,
where $\Delta_{\mathrm{MDE}} = \delta_{\mathrm{MDE}}/\sigma$.

\paragraph{Sequential boundaries.}
We consider three sequential testing boundaries. The burn-in
parameter $m$ is the smallest \emph{total} sample size at which the
confidence sequence is evaluated. In practice, $m$ is a small
constant (e.g.\ $m = 20$ total, giving $10$ per arm at $r = 1$)
chosen so that the central limit approximation is reliable, typically
$m \ll \nz$. The value of $m$ affects both the boundary shape and the
monitoring window; Section~\ref{sec:sensitivity} studies its impact
on $\kstar$. For general $r$, define the standardised cumulative process
\begin{equation}
  Z_n := \frac{\hat\delta_n \sqrt{r\,n}}{\sigma\,(1+r)},
  \qquad
  \hat\delta_n = \hat\mu_{T,n} - \hat\mu_{C,n},
  \label{eq:Zn}
\end{equation}
where $n$ is the total sample size. Since
$\mathrm{Var}(\hat\delta_n) = \sigma^2(1+r)^2/(rn)$, we have
$\mathrm{Var}(Z_n) = 1$ for every $r$. At $r = 1$ this reduces to
$Z_n = \sqrt{n}\,\hat\delta_n/(2\sigma)$. The rejection event by
sample size $n$ is $\{\tau \le n\}$ where
$\tau := \inf\{n \ge m : Z_n > f(n)\}$ and $f$ is the boundary
specific to the confidence sequence.

\textit{WSKR boundary.}
\citet[Theorem~3.3]{wskr2023} establish a confidence sequence valid
uniformly over sample sizes $j \ge m_{\mathrm{arm}}$, where $j$ is
the per-arm count. Applied to the $n/2$ paired differences
$D_i = X_{T,i} - X_{C,i}$ at equal allocation, the CI becomes
\begin{equation}
  \bar{C}_n^{(m)}(\alpha)
  = \hat\delta_n \pm \hat\sigma_{D,n} \,
    \sqrt{\frac{\La + \log(n/m)}{n/2}},
  \qquad n \ge m,
  \label{eq:wskrCI}
\end{equation}
where $\hat\sigma_{D,n}$ is the sample standard deviation of the
paired differences and
$\La := \Psi^{-1}(1-\alpha)$ with $\Psi$ the CDF of the
Robbins-Siegmund limiting distribution \citep{robbinssiegmund1970}
($\La = 9.50, 7.67, 6.35, 4.93$ at $\alpha = 0.01, 0.025, 0.05, 0.10$).
We follow the symmetric-CI calibration of \citet{wskr2023}, in which
$\La$ controls the simultaneous two-sided miscoverage at $\alpha$, so
the one-sided test of $H_0: \delta \le 0$ has type-I error at most
$\alpha$.
The closed-form analysis below replaces $\hat\sigma_{D,n}$
with the known $\sigma_D = \sqrt{2}\,\sigma$, moving from the distribution-free
WSKR guarantee to a Gaussian approximation. The boundary in $Z_n$
units is $f_{\mathrm{W}}(n) = \sqrt{\La + \log(n/m)}$.

\textit{Maharaj boundary.}
\citet{maharaj2023} construct an asymptotic confidence sequence from
a Gaussian-mixture supermartingale calibrated at miscoverage $\alpha$,
with tuning constant
$\lambda_{\mathrm{M}}(\alpha) = -W_{-1}(-\alpha^2 / e) - 1$, where
$W_{-1}$ is the lower real branch of the Lambert-$W$ function
\citep{howardetal2021}. The boundary in $Z_n$ units is
\begin{equation}
  f_{\mathrm{M}}(n)
  = \sqrt{\frac{m}{n}\cdot
    \frac{2(\lambda_{\mathrm{M}} n/m + 1)}{\lambda_{\mathrm{M}}}
    \,\log\!\left(1 + \frac{\sqrt{\lambda_{\mathrm{M}} n/m + 1}}
                           {2\alpha}\right)}.
  \label{eq:fM}
\end{equation}

\textit{mSPRT boundary.}
The mixture sequential probability ratio test
\citep{johariopre2022} averages the likelihood ratio over a Gaussian
mixing distribution $N(0, \sigma_\tau^2)$ on the effect size,
producing a test statistic that coincides with the Bayes factor under
the same prior. Setting the mixing standard deviation to the MDE,
$\sigma_\tau = \delta_{\mathrm{MDE}}$, with
$\kappa_n = \sigma_\tau^2 / \mathrm{Var}(\hat\delta_n)
= \Delta_{\mathrm{MDE}}^2\, r\,n / (1+r)^2$,
the symmetric rejection rule $\mathrm{BF}_n > 1/\alpha$, equivalent
to $|Z_n| > f_{\mathrm{J}}(n)$, gives the boundary
\begin{equation}
  f_{\mathrm{J}}(n) = \sqrt{\frac{2(1+\kappa_n)}{\kappa_n}
  \left(\log\frac{1}{\alpha}
  + \tfrac{1}{2}\log(1+\kappa_n)\right)}.
  \label{eq:fJ}
\end{equation}
For the one-sided test $H_0: \delta \le 0$, we use the positive
branch $Z_n > f_{\mathrm{J}}(n)$. Since
$\{\exists n: Z_n > f_{\mathrm{J}}(n)\} \subseteq
\{\exists n: |Z_n| > f_{\mathrm{J}}(n)\}$, the one-sided type-I
error is at most $\alpha$. For the composite null $\delta \le 0$,
the crossing probability under any $\delta < 0$ is strictly smaller
than at $\delta = 0$ because $Z_n$ under $\delta < 0$ is
stochastically dominated by $Z_n$ under $\delta = 0$.
Unlike the WSKR and Maharaj boundaries, $f_{\mathrm{J}}$ depends on
the effect size through $\sigma_\tau$. When $\sigma_\tau = \delta_{\mathrm{MDE}}$
this dependence is absorbed into $\nz$. The closed-form $\kstarJ$
depends on $t_0$ through the monitoring window, but the dependence
is weak: $\kstarJ$ varies by less than $0.5\%$ across the range
$t_0 \in [0.001, 1]$. Because the positive branch of
$\mathrm{BF}_n > 1/\alpha$ is a stopping rule for declaring a
positive effect in Bayesian A/B testing \citep{johariopre2022},
$\kstarJ$ applies directly to sample-size planning for that decision.

\paragraph{Boundary in rescaled time.}
Rescaling normalises the fixed-sample endpoint to $t = 1$ and makes
the boundary crossing problem amenable to Brownian-motion analysis.
We work in the rescaled time variable $t = n/\nz$ and the rescaled
cumulative process $Y_t = \sqrt{t}\, Z_{t \nz}$. Write
$t_0 := m/\nz$ for the burn-in in rescaled time. Under $H_1$ with
$\delta = \delta_{\mathrm{MDE}}$ and the moment conditions of
\citet{wskr2023}, $Y_t$ is asymptotically Brownian motion with drift
$\mu_d = z_\alpha + z_\beta$ per unit $t$, started at
$Y_{t_0} \sim N(t_0 \mu_d, t_0)$
(Appendix~\ref{app:derivation}). Because $\mathrm{Var}(Z_n) = 1$
for any $r$ by \eqref{eq:Zn}, the rescaled process $Y_t$ has drift
$\mu_d$ and variance rate one regardless of $r$; the allocation ratio
is absorbed into $\nz$ via \eqref{eq:nz}. The correction factor
$\kstar(\alpha, \beta, t_0)$ therefore depends on $r$ only through
$t_0 = m/\nz(r)$; at equal $t_0$, $\kstar$ is identical across $r$.
The WSKR boundary becomes
\begin{equation}
  b_{\mathrm{W}}(t) = \sqrt{t\,(\La + \log(t/t_0))},
  \label{eq:bW}
\end{equation}
and the Maharaj boundary becomes
\begin{equation}
  b_{\mathrm{M}}(t)
  = \sqrt{t_0}\, g_{\mathrm{M}}(t/t_0),
  \label{eq:bM}
\end{equation}
where $g_{\mathrm{M}}(v)^2 = (2(\lambda_{\mathrm{M}} v + 1) /
\lambda_{\mathrm{M}})\, \log(1 + \sqrt{\lambda_{\mathrm{M}} v + 1}
/ (2\alpha))$.
With $\sigma_\tau = \delta_{\mathrm{MDE}}$, the mSPRT boundary becomes
\begin{equation}
  b_{\mathrm{J}}(t)^2
  = \frac{2(1 + t\mu_d^2)}{\mu_d^2}\,
    \left(\log\frac{1}{\alpha}
    + \tfrac{1}{2}\log(1 + t\mu_d^2)\right),
  \label{eq:bJ}
\end{equation}
where $\mu_d = z_\alpha + z_\beta$ (since
$\kappa_{t\nz} = \Delta^2 r\,t\,\nz/(1+r)^2 = t\mu_d^2$ by \eqref{eq:nz}).
The mSPRT boundary does not depend on $m$.
All three boundaries are concave on the relevant range. For
$b_{\mathrm{W}}$, $b_{\mathrm{W}}''(t) < 0$ on $[t_0, \infty)$
(Appendix~\ref{app:derivation}, A.1). For $b_{\mathrm{M}}$,
$b_{\mathrm{M}}''(t) < 0$ on $[t_0, \infty)$ for all
$\lambda_{\mathrm{M}} > 0$ and $\alpha \in (0, \tfrac{1}{2})$
(Appendix~\ref{app:derivation}, A.1). For $b_{\mathrm{J}}$,
concavity is verified numerically on $[t_0, k]$ at every
$(\alpha, \beta)$ pair in the grid; users applying the formula at
untabulated parameters should verify that $b_{\mathrm{J}}''(t) < 0$
on a grid in $[t_0, k]$.
The closed-form approximation requires concavity
(Section~\ref{sec:result}) and depends on the boundary only through $b(k)$
and $b'(k)$ at the planned endpoint.

\paragraph{Power.}
We write $t_0 = m/\nz$ for the burn-in, $t = 1$ for the
fixed-sample baseline, and $k$ for the planned endpoint, so
$n = k \cdot \nz$. The AV power at horizon $k \cdot \nz$ is
\begin{equation}
  \pi(k; \delta_{\mathrm{MDE}})
  := P(\tau \le k \cdot \nz \mid \delta = \delta_{\mathrm{MDE}}),
  \label{eq:pidef}
\end{equation}
the probability of a boundary crossing at any time from burn-in $m$
to the planned endpoint $k \cdot \nz$. With $m \ll \nz$, the
monitoring window $[m,\, k \cdot \nz]$ extends well before the
fixed-sample size.

\paragraph{Last-point sizing.}
The last-point rule sets $n = k \cdot \nz$ to the smallest value at
which the marginal rejection probability at the endpoint reaches
$1 - \beta$. Under $H_1$, $Y_k \sim N(k \mu_d, k)$, so this
probability is $\Phi\!\big((k \mu_d - b(k)) / \sqrt{k}\big)$.
Setting it equal to $1-\beta$ gives
$b(k) = \sqrt{k}\,(\sqrt{k}\,\mu_d - z_\beta)$.
For the WSKR boundary this reduces to
$\sqrt{\La + \log(k/t_0)} = \sqrt{k}\,\mu_d - z_\beta$,
and for each boundary the root $\knaive$ is found numerically.
Last-point sizing is conservative by construction, as the probability
of rejecting before the last planned time point is ignored.

\section{The corrected factor}
\label{sec:result}

The corrected factor is the smallest $k > t_0$ at which a tangent-line
approximation to the AV power \eqref{eq:pidef} equals $1 - \beta$.
Under the Brownian-motion approximation to the cumulative process
under $H_1$, this approximation depends on the boundary $b$ only
through $b(k)$ and $b'(k)$; accuracy depends on how closely the
tangent tracks the boundary on $[t_0, k]$ and is checked by simulation
in Section~\ref{sec:simulation}. Let
$\mu_d := z_\alpha + z_\beta$ and define the tangent of $b$ at the
planned endpoint $t = k$:
\begin{equation}
  L_{t_0}(k) := b(k) - b'(k)(k - t_0),
  \qquad
  s(k) := b'(k).
\end{equation}
Under concavity, the tangent lies above the boundary on $[t_0, k]$.
Set $T := k - t_0$, $c_x := (L_{t_0} - t_0 \mu_d)/\sqrt{t_0}$,
$\nu := \mu_d - s$, and
$\rho := -\sqrt{t_0}/\sqrt{T+t_0}$. Then
\begin{equation}
  \pi(k; \delta_{\mathrm{MDE}}) \approx \Phi(-c_x) + I_1(k) + I_2(k),
  \label{eq:piclosed}
\end{equation}
with
\begin{align}
  I_1(k) &= \Phi_2\!\left(c_x,\ \tfrac{\nu T - \sqrt{t_0}\,c_x}{\sqrt{T+t_0}};\ \rho\right),
  \label{eq:I1}\\
  I_2(k) &= e^{2\sqrt{t_0}\,\nu\,(c_x + \sqrt{t_0}\,\nu)}\,
           \Phi_2\!\left(c_x + 2\sqrt{t_0}\,\nu,\ \tfrac{-\nu T - \sqrt{t_0}\,c_x - 2t_0\nu}{\sqrt{T+t_0}};\ \rho\right).
  \label{eq:I2}
\end{align}
Here $\Phi_2(a, b; \rho)$ is the bivariate standard normal CDF with
correlation $\rho$.
The derivation (Appendix~\ref{app:derivation}) uses tangent
linearisation, the Bachelier first-passage formula, and Gaussian
completion-of-the-square to integrate over the burn-in initial value
$Y_{t_0} \sim N(t_0 \mu_d, t_0)$. The three terms have plain
interpretations:
$\Phi(-c_x)$ is the immediate-rejection probability at $t = t_0$
for the tangent-line surrogate,
$I_1$ is the probability of crossing the linearised boundary during
$(t_0, k]$, and $I_2$ is the Bachelier reflection correction. The
closed-form $\kstar$ is the smallest $k > t_0$ solving
$\pi_{\mathrm{closed}}(k) = 1-\beta$, where $\pi_{\mathrm{closed}}$
is the right-hand side of \eqref{eq:piclosed}. A bracket search
returns $\kstar$ at two bivariate normal CDF calls per
evaluation.

\paragraph{Boundary-specific slopes.}
For the WSKR boundary,
$b_{\mathrm{W}}'(k) = (\La + \log(k/t_0) + 1)\,/\,(2\,b_{\mathrm{W}}(k))$.
For the Maharaj boundary,
$b_{\mathrm{M}}'(k) = g_{\mathrm{M}}'(k/t_0)\,/\,\sqrt{t_0}$,
where $g_{\mathrm{M}}'(v) = (2 q + \xi/(2\alpha + \xi))\,/\,(2\,g_{\mathrm{M}}(v))$
with $\xi = \sqrt{\lambda_{\mathrm{M}} v + 1}$ and
$q = \log(1 + \xi/(2\alpha))$.
For the mSPRT boundary, with $\eta = k\mu_d^2$ and
$c = \log(1/\alpha)$,
$b_{\mathrm{J}}'(k) = (c + \tfrac{1}{2}\log(1+\eta) + \tfrac{1}{2})\,/\,
b_{\mathrm{J}}(k)$.
The closed form satisfies
$\pi_{\mathrm{closed}}(t_0) < 1-\beta$ and
$\pi_{\mathrm{closed}}(k) \to 1$ as $k \to \infty$. A bracket search
finds a unique root $\kstar$ at every $(\alpha, \beta)$ cell;
no second root was observed on the tested grid
(Appendix~\ref{app:derivation}).
Table~\ref{tab:kstar} reports $\kstar$ and the relative saving
$S = 100 \cdot (\knaive - \kstar)/\knaive$ at $m = 40$ and
$\Delta = 0.1$ on the $(\alpha, \beta)$ grid;
Section~\ref{sec:sensitivity} shows how $\kstar$ varies with $m$.
The three boundaries arise from different mathematical constructions,
so the factors should not be read as an efficiency ranking.

\begin{table}[t]
\centering
\small
\begin{tabular}{cc|cc|cc|cc}
\toprule
& & \multicolumn{2}{c|}{WSKR} & \multicolumn{2}{c|}{Maharaj} & \multicolumn{2}{c}{mSPRT} \\
$\alpha$ & $\beta$ & $\kstar$ & saving & $\kstarM$ & saving & $\kstarJ$ & saving \\
\midrule
0.010 & 0.05 & 1.789 & $8.3\%$  & 1.865 & $7.9\%$  & 1.593 & $9.0\%$ \\
0.010 & 0.10 & 1.859 & $9.3\%$  & 1.945 & $8.9\%$  & 1.647 & $9.9\%$ \\
0.010 & 0.20 & 1.964 & $10.8\%$ & 2.068 & $10.3\%$ & 1.732 & $11.3\%$ \\
0.025 & 0.05 & 1.954 & $9.3\%$  & 2.027 & $9.0\%$  & 1.706 & $10.2\%$ \\
0.025 & 0.10 & 2.051 & $10.5\%$ & 2.136 & $10.1\%$ & 1.779 & $11.4\%$ \\
0.025 & 0.20 & 2.202 & $12.4\%$ & 2.306 & $11.8\%$ & 1.897 & $13.1\%$ \\
0.050 & 0.05 & 2.153 & $10.3\%$ & 2.209 & $10.0\%$ & 1.834 & $11.5\%$ \\
0.050 & 0.10 & 2.288 & $11.7\%$ & 2.354 & $11.3\%$ & 1.930 & $12.9\%$ \\
0.050 & 0.20 & 2.504 & $13.8\%$ & 2.588 & $13.4\%$ & 2.092 & $15.0\%$ \\
0.100 & 0.05 & 2.450 & $11.5\%$ & 2.489 & $11.3\%$ & 2.033 & $13.2\%$ \\
0.100 & 0.10 & 2.652 & $13.2\%$ & 2.700 & $12.9\%$ & 2.172 & $15.0\%$ \\
0.100 & 0.20 & 2.993 & $15.9\%$ & 3.059 & $15.5\%$ & 2.419 & $17.5\%$ \\
\bottomrule
\end{tabular}
\caption{Closed-form sample-size factor and relative saving over the
last-point rule on the $(\alpha, \beta)$ grid at $m = 40$
and $\Delta = 0.1$ ($t_0 \approx 0.02$). All three $\kstar$ values
come from \eqref{eq:piclosed} with the corresponding $b(k)$ and
$b'(k)$ substituted. The mSPRT boundary does not depend on $m$
(Section~\ref{sec:setup}).}
\label{tab:kstar}
\end{table}

\section{Simulation}
\label{sec:simulation}

We validate $\kstar$, $\kstarM$, and $\kstarJ$ by direct Monte Carlo
simulation of all three boundaries applied to two-sample experiments. The
principal parameters are $\alpha = 0.05$, $1-\beta = 0.80$,
burn-in $m = 20$, and $B = 50{,}000$ replications per cell.
Full data-generating process and stopping rule are in
Appendix~\ref{app:simdesign}.
Simulation code is available at
\url{https://github.com/MSchultzberg/sequential-ssc-correction-approximation}.
Before validating power, we verify type-I error control by running
$B = 200{,}000$ replications under $H_0$ ($\delta = 0$): the
empirical one-sided rejection rate is below $\alpha$ for all three
boundaries at every cell in the grid.

\subsection{Empirical power}
\label{sec:simpower}

Table~\ref{tab:simpower} reports the empirical AV power
\begin{equation*}
  \hat\pi(n) = \frac{1}{B} \sum_{\ell=1}^{B}
  \mathbf{1}\!\big[\tau^{(\ell)} \le n\big]
\end{equation*}
at $\alpha = 0.05$, $1 - \beta = 0.80$ under the last-point and
corrected sample sizes for all three boundaries.

\begin{table}[ht]
\centering
\small
\begin{tabular}{lc cccccc}
\toprule
& & \multicolumn{2}{c}{WSKR} & \multicolumn{2}{c}{Maharaj} & \multicolumn{2}{c}{mSPRT} \\
\cmidrule(lr){3-4} \cmidrule(lr){5-6} \cmidrule(lr){7-8}
Outcome & $\Delta$ &
$\hat\pi_{\knaive}$ & $\hat\pi_{\kstar}$ &
$\hat\pi_{\knaiveM}$ & $\hat\pi_{\kstarM}$ &
$\hat\pi_{\knaiveJ}$ & $\hat\pi_{\kstarJ}$ \\
\midrule
Gaussian, $r{=}1$ & $0.1$ & $0.880$ & $0.826$ & $0.879$ & $0.816$ & $0.877$ & $0.806$ \\
Gaussian, $r{=}1$ & $0.2$ & $0.885$ & $0.819$ & $0.879$ & $0.815$ & $0.875$ & $0.803$ \\
Gaussian, $r{=}1$ & $0.3$ & $0.883$ & $0.816$ & $0.880$ & $0.813$ & $0.869$ & $0.801$ \\
Gaussian, $r{=}1$ & $0.5$ & $0.879$ & $0.814$ & $0.874$ & $0.804$ & $0.867$ & $0.797$ \\
\midrule
Gaussian, $r{=}1.5$ & $0.2$ & $0.884$ & $0.821$ & $0.879$ & $0.814$ & $0.873$ & $0.804$ \\
Gaussian, $r{=}2$ & $0.2$ & $0.883$ & $0.820$ & $0.880$ & $0.812$ & $0.874$ & $0.803$ \\
\bottomrule
\end{tabular}
\caption{Empirical AV power at $\alpha = 0.05$, $1-\beta = 0.80$,
$m = 20$, $B = 50{,}000$ replications per cell,
Gaussian DGP. The $r{=}1.5$ and $r{=}2$ rows use the same $\kstar$
as $r{=}1$ (computed at the $r$-specific $t_0$); empirical power is
within MC error of the $r{=}1$ values. Monte Carlo standard error
is below $0.0023$ at every cell.}
\label{tab:simpower}
\end{table}

Empirical power for Gaussian outcomes is close to and usually above
$1-\beta$, though finite-sample effects (CLT error and discrete
monitoring) can reduce it slightly below target at individual cells.
The tendency to overshoot $1-\beta$ has a geometric explanation:
under concavity, the tangent lies above the boundary on $[t_0, k]$,
so the linearised crossing probability is a lower bound on the true
Brownian crossing probability and $\kstar$ is slightly larger than
needed. The surplus under the last-point factors is a property of
last-point sizing, not of the boundary or the outcome distribution.

\subsection{Sample-size saving}

The relative saving of $\kstar$ over $\knaive$ is
\begin{equation}
  S(\alpha, \beta)
  := 100 \cdot \frac{\knaive(\alpha, \beta) - \kstar(\alpha, \beta)}
                      {\knaive(\alpha, \beta)}
  \quad \%.
  \label{eq:saving}
\end{equation}
Because $\kstar$ depends on $t_0$, the saving $S$ is not a single number
but varies with the experiment's effect size and burn-in.

The saving varies with the effect size because $t_0 = m/\nz$ depends
on $\Delta$: smaller effects give larger $\nz$ and hence smaller
$t_0$, wider boundaries, and lower savings.
Over the $(\alpha, \beta)$ grid (Table~\ref{tab:kstar}) and the
sensitivity analysis (Table~\ref{tab:sensitivity}), the saving ranges
from $8\%$ to $20\%$ across all three boundaries and the full range of
burn-in values. Full simulation design is in
Appendix~\ref{app:simdesign}.

\subsection{Production validation}

Figure~\ref{fig:prodsaving} shows the saving on $713$ metrics from the last 283
experiments using always-valid inference, taken from Confidence, Spotify's
commercial experimentation platform. The platform uses the Maharaj boundary with
Bonferroni corrections for multiple comparisons and multiple metrics, which is
why the production validation is for that boundary only. The median saving is
$9.5\%$ and the mean is $9.8\%$. The realised saving is lower than the $13.4\%$
at the $(\alpha, \beta) = (0.05, 0.20)$ entry of Table~\ref{tab:kstar}
because Bonferroni adjustment pushes the effective $\alpha_{\mathrm{adj}}$ and
$\beta_{\mathrm{adj}}$ toward the lower-saving corner of the grid. See \citet{schultzberg2026risk} and \citet{schultzberg2026bonf}
for details on the multiple-metric handling.

\begin{figure}[t]
\centering
\includegraphics[width=\linewidth]{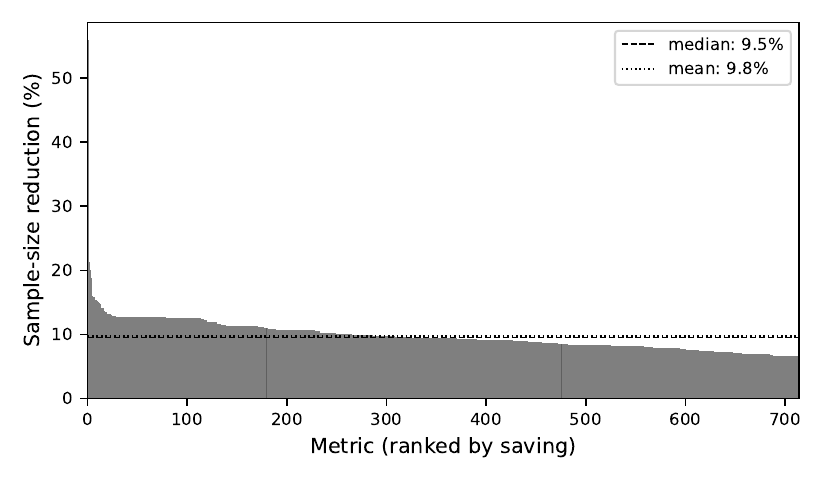}
\caption{Sample-size reduction from the corrected factor $\kstarM$
relative to the last-point rule across $713$ metrics from Spotify's
experimentation platform (Maharaj boundary). Each bar is one metric,
ranked by saving. Bonferroni corrections for multiple comparisons and
metrics shift the effective $(\alpha, \beta)$ toward the lower-saving
region of Table~\ref{tab:kstar}.}
\label{fig:prodsaving}
\end{figure}

\subsection{Sensitivity to burn-in}
\label{sec:sensitivity}

The burn-in $m$ controls where monitoring begins and affects $\kstar$
through two channels: the monitoring window $[t_0, k]$ widens as $m$
decreases, and the WSKR and Maharaj boundaries shift because $t_0$
enters the boundary formula.
Table~\ref{tab:sensitivity} reports $\kstar$ at
$(\alpha, 1-\beta) = (0.05, 0.80)$ for a range of total burn-in
values from $m = 20$ to $m = \nz$, where $m = \nz$ corresponds to
monitoring only from the fixed-sample size onward.

\begin{table}[ht]
\centering
\small
\begin{tabular}{rr|ccc|ccc}
\toprule
& & \multicolumn{3}{c|}{$\kstar$} & \multicolumn{3}{c}{saving (\%)} \\
$m$ & $t_0$ & WSKR & Maharaj & mSPRT & WSKR & Maharaj & mSPRT \\
\midrule
\multicolumn{8}{c}{$\Delta = 0.05$ \quad ($\nz = 9{,}892$)} \\
\midrule
   $20$  & $0.002$ & $2.96$ & $3.04$ & $2.09$ & $11.7$ & $11.4$ & $15.0$ \\
  $100$  & $0.010$ & $2.61$ & $2.69$ & $2.09$ & $13.3$ & $12.8$ & $15.0$ \\
  $200$  & $0.020$ & $2.45$ & $2.54$ & $2.09$ & $14.1$ & $13.6$ & $15.0$ \\
 $1000$  & $0.101$ & $2.09$ & $2.19$ & $2.09$ & $16.6$ & $15.7$ & $15.0$ \\
 $2000$  & $0.202$ & $1.92$ & $2.05$ & $2.09$ & $18.0$ & $16.6$ & $15.0$ \\
$\nz$   & $1.0$  & $1.59$ & $1.82$ & $2.10$ & $19.3$ & $16.5$ & $14.6$ \\
\midrule
\multicolumn{8}{c}{$\Delta = 0.1$ \quad ($\nz = 2{,}474$)} \\
\midrule
   $20$  & $0.008$ & $2.66$ & $2.74$ & $2.09$ & $13.0$ & $12.6$ & $15.0$ \\
   $40$  & $0.016$ & $2.50$ & $2.59$ & $2.09$ & $13.8$ & $13.4$ & $15.0$ \\
  $100$  & $0.040$ & $2.30$ & $2.39$ & $2.09$ & $15.1$ & $14.5$ & $15.0$ \\
  $200$  & $0.081$ & $2.14$ & $2.24$ & $2.09$ & $16.2$ & $15.4$ & $15.0$ \\
 $1000$  & $0.404$ & $1.76$ & $1.92$ & $2.09$ & $19.7$ & $17.4$ & $15.0$ \\
$\nz$   & $1.0$  & $1.59$ & $1.82$ & $2.10$ & $19.3$ & $16.5$ & $14.6$ \\
\bottomrule
\end{tabular}
\caption{Sensitivity of $\kstar$ to burn-in at $(\alpha, 1-\beta) = (0.05, 0.80)$.
As $m$ decreases, the WSKR and Maharaj boundaries widen (via the
$\log(n/m)$ term), increasing $\kstar$. The mSPRT boundary does not
depend on $m$, so $\kstarJ$ is essentially constant. The last row
($m = \nz$, $t_0 = 1$) recovers the convention of monitoring
only from the fixed-sample size onward.}
\label{tab:sensitivity}
\end{table}

At realistic burn-ins ($m \le 200$), the WSKR and Maharaj
factors are \emph{larger} than at $m = \nz$ because the
$\log(n/m)$ term in the boundary widens the confidence sequence when
$m$ is small. The wider boundary more than offsets the longer
monitoring window, so $\kstar$ increases as $m$ decreases. Despite
this, the \emph{saving} over the last-point rule remains substantial
(11--20\%) across the range. The mSPRT boundary does not depend on
$m$; $\kstarJ$ depends on $m$ only through the monitoring window
and varies by less than $0.5\%$ across the operating range.
The practical implication: the total sample budget $\kstar \cdot \nz$
for a confidence-sequence-based sequential test with realistic burn-in is $2$--$3$
times the fixed-sample size, considerably more than the factor of approximately $1.6$ for the WSKR
boundary at $m = \nz$ (last row of Table~\ref{tab:sensitivity}). The
correction from the last-point rule remains valuable but the baseline
cost of sequential inference with realistic burn-in is higher.

\section{Conclusion}

The last-point rule for sequential experiments oversizes by seven to
nine percentage points at standard operating characteristics. The
closed-form factor $\kstar$ approximately corrects this via tangent
linearisation, verified across three boundary families and an
extended $(\alpha, \beta)$ grid. The correction saves $8$--$20\%$
of the last-point sample budget across the operating range.
At low base rates ($p_C \le 0.01$), the CLT approximation degrades
and power falls below target (Table~\ref{tab:lowbase}).

With realistic burn-in ($m \ll \nz$), the total sample budget for a
confidence-sequence-based sequential test is $2$--$3$ times the
fixed-sample size, reflecting the cost of the wider boundary that
validity over $[m, \infty)$ requires.

The production validation on Spotify's experimentation platform
shows a median sample-size saving of $9.5\%$ across $713$ metrics.

\bibliography{paper}

\appendix

\section{Derivation of $\kstar$}
\label{app:derivation}

The derivation strategy (tangent linearisation of a curved boundary,
Bachelier first-passage on the linear surrogate, integration over the
initial value) follows \citet{siegmund1977}, applied here to the
specific boundary families arising in modern confidence sequences.

The proof uses the WSKR boundary as the worked example; the Maharaj
boundary follows by substituting $b_{\mathrm{M}}(k)$ and
$b_{\mathrm{M}}'(k)$ for $b_{\mathrm{W}}(k)$ and
$b_{\mathrm{W}}'(k)$ in the final expressions. The derivation holds
for any allocation ratio $r \ge 1$: with
$Z_n = \hat\delta_n\sqrt{r\,n}/(\sigma(1+r))$ we have
$\mathrm{Var}(Z_n) = 1$ for every $r$, and under $H_1$ the rescaled
process $Y_t = \sqrt{t}\,Z_{t\nz}$ satisfies
$E(Y_t) = \mu_d\,t$ and $\mathrm{Var}(Y_t) = t$ since
$E(Z_n) = \delta\sqrt{r\,n}/(\sigma(1{+}r))$ and
$n = t\,\nz$ absorbs $(1{+}r)^2/r$ through \eqref{eq:nz}. The
boundary in rescaled time, the Bachelier passage formula, and the
integration over $Y_{t_0}$ therefore depend on $r$ only through
$t_0 = m/\nz(r)$. At equal $t_0$, $\kstar$ is identical across $r$.

\paragraph{A.1. Process and boundary on the rescaled time axis.}
The proof works on the rescaled time variable $t = n/\nz$ and the
rescaled cumulative process
\begin{equation}
  Y_t := \sqrt{t}\, Z_{t \nz},
\end{equation}
where $Z_n = \hat\delta_n \sqrt{r\,n}/(\sigma(1{+}r))$ has unit
variance for every $r$ by \eqref{eq:Zn}. Therefore
$\mathrm{Var}(Y_t) = t \cdot \mathrm{Var}(Z_{t\nz}) = t$
on the rescaled axis; $Y$ has variance rate one in $t$.
Increments of $Y$ over disjoint $t$-intervals correspond to
non-overlapping blocks of observations, so they are
asymptotically independent under the i.i.d.\ assumption. Under $H_1$
with $\delta = \delta_{\mathrm{MDE}}$ and the moment conditions of
\citet{wskr2023}, $Y_t$ converges weakly to Brownian motion with constant
drift $\mu_d = z_\alpha + z_\beta$ per unit $t$
\citep[by the functional CLT;][Theorem~14.1]{billingsley1968}, with
$Y_{t_0} \sim N(t_0\mu_d, t_0)$ at the burn-in $t_0 = m/\nz$. The WSKR boundary
$f(j) = \sqrt{\La + \log(j/m)}$ acting on $\{Z_n\}$ transforms onto
$\{Y_t\}$ via
\begin{equation}
  Y_t > \sqrt{t}\, f(t \nz)
  = \sqrt{t} \cdot \sqrt{\La + \log(t/t_0)}
  = \sqrt{t (\La + \log(t/t_0))} =: b_{\mathrm{W}}(t),
  \qquad t \ge t_0.
  \label{eq:bdytransform}
\end{equation}
Writing $u = \La + \log(t/t_0)$, so $b_{\mathrm{W}}(t) = \sqrt{t\, u}$ and
$du/dt = 1/t$, the chain rule gives
\begin{equation*}
  b_{\mathrm{W}}'(t) = \frac{u + 1}{2\, b_{\mathrm{W}}(t)}
        = \frac{\La + \log(t/t_0) + 1}{2\sqrt{t\,(\La + \log(t/t_0))}},
  \qquad
  b_{\mathrm{W}}''(t) = -\frac{u^2 + 1}{4\, b_{\mathrm{W}}(t)^3} < 0
  \quad \text{on } [t_0, \infty),
\end{equation*}
so $b_{\mathrm{W}}$ is concave there.

For the Maharaj boundary, write $g(t) = b_{\mathrm{M}}(t)^2$ and
$\varrho = \sqrt{\lambda_{\mathrm{M}} t/t_0 + 1}\,/\,(2\alpha)$.
Since $b_{\mathrm{M}} = \sqrt{g}$, the sign of $b_{\mathrm{M}}''$
equals the sign of $2g\,g'' - (g')^2$. After substitution, this
reduces to showing $\varphi(\varrho) < 0$ for $\varrho > 0$, where
\[
  \varphi(\varrho)
  = \frac{2\varrho\log(1{+}\varrho) - \varrho^2}{(1{+}\varrho)^2}
    - 4\log^2(1{+}\varrho).
\]
The bound $\log(1{+}\varrho) < \varrho$ gives
$2\varrho\log(1{+}\varrho) - \varrho^2 < \varrho^2$, so the first
term is below $\varrho^2/(1{+}\varrho)^2$. The bound
$\varrho/(1{+}\varrho) < 2\log(1{+}\varrho)$ (verified by comparing
derivatives at $\varrho = 0$) yields
$\varrho^2/(1{+}\varrho)^2 < 4\log^2(1{+}\varrho)$, so
$\varphi(\varrho) < 0$ and $b_{\mathrm{M}}$ is strictly concave on
$[t_0, \infty)$.

\paragraph{A.2. Linearisation.}
Replace $b$ on $[t_0, k]$ by its tangent at $t = k$:
\begin{equation}
  L(t) = b(k) + s(t - k),
  \qquad
  s = b'(k),
  \qquad
  L_{t_0} = b(k) - s(k - t_0).
\end{equation}

\paragraph{A.3. Bachelier first-passage on the linear boundary.}
The classical Bachelier formula for Brownian motion $W_t$ with drift
$\mu$ and variance rate $1$ states that for a linear boundary
$\ell(t) = a + \mu_b t$ on $[0, T]$, with $W_0 = 0$ and $a > 0$,
\begin{equation*}
  P\!\left(\sup_{0 \le t \le T}\big(W_t - \ell(t)\big) \ge 0\right)
  = \Phi\!\left((\mu - \mu_b)\sqrt{T} - \tfrac{a}{\sqrt{T}}\right)
  + e^{-2 a (\mu_b - \mu)}\,
    \Phi\!\left(-(\mu - \mu_b)\sqrt{T} - \tfrac{a}{\sqrt{T}}\right),
\end{equation*}
\citep{bachelier1900,siegmund1985}. Substituting $\mu = \mu_d$ and
$\mu_b = s$ flips the sign in the reflection exponent since
$-(\mu_b - \mu) = \mu_d - s$. Translating time so that $t = t_0$
corresponds to time $0$ in the standard Bachelier statement, with horizon
$T = k - t_0$, drift $\mu = \mu_d$, and linear boundary slope
$\mu_b = s$, the first-passage probability of $Y$ over $L$ on $[t_0, k]$
conditional on $Y_{t_0} = x < L_{t_0}$ is
\begin{equation}
  P(\tau \le k \mid Y_{t_0} = x)
  = \Phi\!\left((\mu_d - s)\sqrt{T} - \tfrac{a}{\sqrt T}\right)
  + e^{2 a (\mu_d - s)}\,
    \Phi\!\left(-(\mu_d - s)\sqrt{T} - \tfrac{a}{\sqrt T}\right),
  \label{eq:bachelier}
\end{equation}
with $a = L_{t_0} - x > 0$. For $x \ge L_{t_0}$ the trajectory
already lies above the linearised boundary at $t = t_0$ and the
conditional probability is one.

\paragraph{A.4. Integration over the burn-in initial value.}
Substitute $u = (Y_{t_0} - t_0\mu_d)/\sqrt{t_0} \sim N(0, 1)$ with
standard normal density $\phi$. With
$c_x = (L_{t_0} - t_0\mu_d)/\sqrt{t_0}$ and $\nu = \mu_d - s$, the
event $Y_{t_0} \ge L_{t_0}$ becomes $\{u \ge c_x\}$ and contributes
$P(u \ge c_x) = \Phi(-c_x)$. The
remaining integral is
\begin{equation}
  \int_{-\infty}^{c_x} P(\tau \le k \mid Y_{t_0} = t_0\mu_d + \sqrt{t_0}\,u)\,
  \phi(u)\, du
  = I_1 + I_2,
\end{equation}
where the two pieces correspond to the two terms of
\eqref{eq:bachelier} with $a = \sqrt{t_0}\,(c_x - u)$. The reduction
uses the following identity.

\begin{lemma}[Gaussian convolution]
\label{lem:Jidentity}
For real $A_0$, $B_0$, and $c$,
\begin{equation}
  \int_{-\infty}^{c} \Phi(A_0 + B_0 u)\, \phi(u)\, du
  = \Phi_2\!\left(c,\ \tfrac{A_0}{\sqrt{1 + B_0^2}};\
                  -\tfrac{B_0}{\sqrt{1+B_0^2}}\right).
  \label{eq:Jidentity}
\end{equation}
\end{lemma}

\begin{proof}
Let $V \sim N(0, 1)$ independent of $u$. Then
$\Phi(A_0 + B_0 u) = P(V \le A_0 + B_0 u \mid u) = P(V - B_0 u \le A_0
\mid u)$. The pair $(u, V - B_0 u)$ is jointly Gaussian with
$\mathrm{Var}(V - B_0 u) = 1 + B_0^2$ and
$\mathrm{Cov}(u, V - B_0 u) = -B_0$. Standardising the second
component gives correlation $-B_0 / \sqrt{1 + B_0^2}$, and the joint
event $\{u \le c,\ V - B_0 u \le A_0\}$ becomes $\{u \le c,\
\widetilde W \le A_0/\sqrt{1+B_0^2}\}$ with $\widetilde W$ unit
variance. The right-hand side of \eqref{eq:Jidentity} follows.
\end{proof}

Applying Lemma~\ref{lem:Jidentity} to the first term of
\eqref{eq:bachelier} with
$A_0 = \nu \sqrt T - \sqrt{t_0}\,c_x / \sqrt T$ and
$B_0 = \sqrt{t_0}/\sqrt T$ yields \eqref{eq:I1} after the substitutions
$1 + B_0^2 = (T+t_0)/T$,
$A_0/\sqrt{1+B_0^2} = (\nu T - \sqrt{t_0}\,c_x)/\sqrt{T+t_0}$,
and $-B_0/\sqrt{1+B_0^2} = -\sqrt{t_0}/\sqrt{T+t_0} = \rho$.

For the second term the integrand carries an extra factor
$e^{2\sqrt{t_0}(c_x - u)\nu}$. Completing the square in the Gaussian
density gives the algebraic identity
\begin{equation}
  e^{-2\sqrt{t_0}\,\nu u}\, \phi(u) = e^{2 t_0\nu^2}\, \phi(u + 2\sqrt{t_0}\,\nu),
\end{equation}
so the prefactor combines as $e^{2\sqrt{t_0}(c_x - u)\nu}\, \phi(u) =
e^{2\sqrt{t_0}\, c_x \nu + 2 t_0\nu^2}\, \phi(u + 2\sqrt{t_0}\,\nu)
= e^{2\sqrt{t_0}\,\nu(c_x + \sqrt{t_0}\,\nu)}\, \phi(u +
2\sqrt{t_0}\,\nu)$. The substitution $w = u + 2\sqrt{t_0}\,\nu$ shifts
the upper limit from $c_x$ to $c_x + 2\sqrt{t_0}\,\nu$.
Lemma~\ref{lem:Jidentity} with
$A_0 = -\nu \sqrt T - \sqrt{t_0}\,c_x/\sqrt T - 2 t_0\nu/\sqrt T$ and
$B_0 = \sqrt{t_0}/\sqrt T$ then yields \eqref{eq:I2}, with the leading
coefficient $e^{2\sqrt{t_0}\,\nu(c_x + \sqrt{t_0}\,\nu)}$. Adding the
three pieces gives \eqref{eq:piclosed}.

\paragraph{A.5. Existence and uniqueness of $\kstar$.}
At $k = t_0$ we have $T = 0$ and $\rho = -1$, so $\Phi_2(c_x, -c_x;
-1) = 0$ and $I_2(t_0) = 0$, giving
$\pi_{\mathrm{closed}}(t_0) = \Phi((\mu_d t_0 - b(t_0))/\sqrt{t_0})$.
Across $\alpha \in \{0.01, 0.025, 0.05, 0.10\}$ and $\beta \in \{0.05,
0.10, 0.20\}$, $\pi_{\mathrm{closed}}(t_0)$ is strictly less than
$1-\beta$ at every cell for all three boundaries. As $k \to \infty$,
$L_{t_0}(k) \to +\infty$, so $c_x \to +\infty$ and
$\Phi(-c_x) \to 0$. The second argument of $I_1$ grows without bound,
giving $I_1 \to 1$. The $\Phi_2$ factor in $I_2$ decays
exponentially, dominating the prefactor, so $I_2 \to 0$. Thus
$\pi_{\mathrm{closed}}(k) \to 1$. Monotonicity of
$\pi_{\mathrm{closed}}(k)$ in $k$ does not follow automatically
because the tangent surrogate changes with $k$ (the crossing events
for different $k$ are not nested). We evaluate
$\pi_{\mathrm{closed}}$ on a fine grid of $k$ values at every
$(\alpha, \beta)$ cell. At every cell, a bracket search
finds a unique root $\kstar < \knaive$; no second root was observed
on the tested grid. Uniqueness is a numerical observation, not a
theorem; we do not prove monotonicity of $\pi_{\mathrm{closed}}$ in
general.

\paragraph{A.6. Numerical verification.}
The closed form \eqref{eq:piclosed} agrees with the Monte Carlo
reference at the precision of the simulation. On the grid
$\alpha \in \{0.01, 0.025, 0.05, 0.10\}$ and
$\beta \in \{0.05, 0.10, 0.20\}$, the closed-form $\kstar$ matches
the Monte Carlo reference within $1.5\%$ relative error (median
$0.9\%$) for the WSKR boundary and within $0.93\%$ (median $0.64\%$)
for the Maharaj boundary. Relative to the finite-sample discrete
simulations, the closed-form root may be slightly smaller because of
CLT and discrete-monitoring residuals that reduce empirical power.
Verification scripts are in the \texttt{simulations.py} file.

\section{Simulation design}
\label{app:simdesign}

\paragraph{Data-generating process.}
Each Monte Carlo replication generates $n_T = \lfloor n/(1{+}r)
\rceil$ treatment and $n_C = n - n_T$ control observations for total
sample sizes $n$ from $m$ to
$n_{\max} = \lceil \knaive \cdot \nz \rceil$,
under one of three distribution families:
\begin{itemize}
  \item Gaussian. $X_{C,j} \sim N(\mu_C, \sigma^2)$,
        $X_{T,j} \sim N(\mu_C + \delta, \sigma^2)$, with $\sigma = 1$
        and $\delta$ chosen so that the standardised effect
        $\Delta \in \{0.1, 0.2, 0.3, 0.5\}$.
  \item Bernoulli. $X_{C,j} \sim \mathrm{Bernoulli}(p_C)$,
        $X_{T,j} \sim \mathrm{Bernoulli}(p_C + \Delta_p)$, with
        $p_C \in \{0.05, 0.20\}$ and $\Delta_p$ chosen so that
        $\Delta = \Delta_p / \sqrt{p_C(1-p_C)} \in \{0.1, 0.2, 0.3, 0.5\}$.
        The treatment and control arms have unequal variances under
        $H_1$; these rows are robustness checks outside the
        common-variance model.
  \item Log-normal. $X_{C,j} = \exp(\xi_{C,j})$ and
        $X_{T,j} = \exp(\xi_{T,j}) + \delta$ with independent
        $\xi_{C,j}, \xi_{T,j} \sim N(0, 1)$ and $\delta$ chosen to
        give standardised effect $\Delta \in \{0.1, 0.2, 0.3, 0.5\}$
        using the variance of $\exp(\xi)$.
\end{itemize}
Allocation ratios $r \in \{1, 1.5, 2\}$ are simulated by drawing
$n_T = n / (1 + r)$ treatment observations and $n_C = r \cdot n_T$
control observations.

\paragraph{Confidence sequence and stopping rule.}
For each replication, the test statistic $Z_n$ from \eqref{eq:Zn}
is evaluated at every total sample size $n$ at which both arms have
gained at least one new observation (step size $\lceil 1{+}r \rceil$).
The boundary is compared in $Z_n$ units: $Z_n > f(n)$, using the
WSKR, Maharaj, or mSPRT boundary as appropriate, with burn-in
$m = 20$ (Table~\ref{tab:kstar} uses $m = 40$).
The first $n$ at which the bound is exceeded is recorded as
$\tau^{(\ell)}$. The empirical power at horizon $n$ is
$\hat\pi(n) = (1/B) \sum_\ell \mathbf{1}[\tau^{(\ell)} \le n]$.

\paragraph{Parameter grid and replications.}
The grid in Tables~\ref{tab:simpower} fixes
$\alpha = 0.05$ and $1 - \beta = 0.80$ across three distributions,
allocation ratios $r \in \{1, 1.5, 2\}$, and effect sizes
$\Delta \in \{0.1, 0.2, 0.3, 0.5\}$. Each cell uses $B = 50{,}000$
replications. The Monte Carlo standard error on the empirical power
is below $0.0023$ at every cell. The broader
$(\alpha, \beta)$ grid in Section~\ref{sec:result} covers
$\alpha \in \{0.01, 0.025, 0.05, 0.10\}$ and
$\beta \in \{0.05, 0.10, 0.20\}$ at $\Delta = 0.2$ and Gaussian
outcomes, with $\kstar$ obtained from the closed form
$\pi_{\mathrm{closed}}(k) = 1 - \beta$. Monte Carlo is used only to
evaluate empirical power $\hat\pi(k \nz)$ at the resulting
closed-form sample size.

\paragraph{Burn-in.}
All simulations use $m = 20$ unless stated otherwise.
Section~\ref{sec:sensitivity} varies $m$ across
$\{20, 40, 100, 200, 1000, 2000, \nz\}$.

\paragraph{Seed and code.}
Simulation scripts and closed-form verification are in the
\texttt{simulations.py} file. Random seed $2026$.

\section{Extended $(\alpha, \beta)$ grid validation}
\label{app:extended}

Tables~\ref{tab:simpower} fix
$\alpha = 0.05$ and $1 - \beta = 0.80$. To verify that the closed
form does not degrade at extreme parameter values, we run the same
Monte Carlo validation on a wider grid.
For the WSKR boundary,
$\alpha \in \{0.01, 0.025, 0.05, 0.10\}$ (limited by the tabulated
$\La$ values).
For the Maharaj boundary, $\lambda_{\mathrm{M}}(\alpha)$ has a
closed form, so we extend to
$\alpha \in \{0.001, 0.005, 0.01, 0.025, 0.05, 0.10\}$.
Both grids use $\beta \in \{0.05, 0.10, 0.20\}$, covering target
power from $0.80$ to $0.95$. All cells use Gaussian outcomes at
$\Delta = 0.2$, $r = 1$, and $B = 50{,}000$ replications per cell.

Table~\ref{tab:extended_wskr} reports results for the WSKR boundary
and Table~\ref{tab:extended_maharaj} for the Maharaj boundary. The
gap column is $\hat\pi(\kstar) - (1-\beta)$. The closed form hits
target power within approximately $3$ percentage points across the
entire grid for all three boundaries.
Power accuracy does not degrade at $\alpha = 0.001$ or at
$1-\beta = 0.95$, though the saving is smaller at extreme $\alpha$
(e.g.\ $6.5\%$ at $\alpha = 0.001$, $1-\beta = 0.95$ for Maharaj).

\begin{table}[ht]
\centering
\small
\begin{tabular}{cccccccc}
\toprule
$\alpha$ & $1-\beta$ & $\knaive$ & $\kstar$ & saving &
$\hat\pi(\knaive)$ & $\hat\pi(\kstar)$ & gap \\
\midrule
$0.010$ & $0.95$ & $1.885$ & $1.722$ & $8.7\%$  & $0.971$ & $0.951$ & $+0.001$ \\
$0.010$ & $0.90$ & $1.975$ & $1.783$ & $9.7\%$  & $0.939$ & $0.906$ & $+0.006$ \\
$0.010$ & $0.80$ & $2.113$ & $1.873$ & $11.4\%$ & $0.870$ & $0.807$ & $+0.007$ \\
$0.025$ & $0.95$ & $2.073$ & $1.870$ & $9.8\%$  & $0.972$ & $0.953$ & $+0.003$ \\
$0.025$ & $0.90$ & $2.198$ & $1.954$ & $11.1\%$ & $0.940$ & $0.908$ & $+0.008$ \\
$0.025$ & $0.80$ & $2.396$ & $2.082$ & $13.1\%$ & $0.879$ & $0.819$ & $+0.019$ \\
$0.050$ & $0.95$ & $2.299$ & $2.050$ & $10.8\%$ & $0.975$ & $0.957$ & $+0.007$ \\
$0.050$ & $0.90$ & $2.471$ & $2.165$ & $12.4\%$ & $0.946$ & $0.913$ & $+0.013$ \\
$0.050$ & $0.80$ & $2.755$ & $2.349$ & $14.8\%$ & $0.886$ & $0.821$ & $+0.021$ \\
$0.100$ & $0.95$ & $2.637$ & $2.314$ & $12.3\%$ & $0.977$ & $0.958$ & $+0.008$ \\
$0.100$ & $0.90$ & $2.895$ & $2.486$ & $14.2\%$ & $0.952$ & $0.918$ & $+0.018$ \\
$0.100$ & $0.80$ & $3.346$ & $2.772$ & $17.2\%$ & $0.893$ & $0.831$ & $+0.031$ \\
\bottomrule
\end{tabular}
\caption{Extended grid validation, WSKR boundary. Gaussian DGP,
$\Delta = 0.2$, $r = 1$, $m = 20$, $B = 50{,}000$. MC SE is
below $0.002$ at every cell.}
\label{tab:extended_wskr}
\end{table}

\begin{table}[ht]
\centering
\small
\begin{tabular}{cccccccc}
\toprule
$\alpha$ & $1-\beta$ & $\knaiveM$ & $\kstarM$ & saving &
$\hat\pi(\knaiveM)$ & $\hat\pi(\kstarM)$ & gap \\
\midrule
$0.001$ & $0.95$ & $1.699$ & $1.590$ & $6.5\%$  & $0.967$ & $0.951$ & $+0.001$ \\
$0.001$ & $0.90$ & $1.757$ & $1.632$ & $7.1\%$  & $0.931$ & $0.904$ & $+0.004$ \\
$0.001$ & $0.80$ & $1.843$ & $1.693$ & $8.1\%$  & $0.856$ & $0.804$ & $+0.004$ \\
$0.005$ & $0.95$ & $1.861$ & $1.719$ & $7.6\%$  & $0.969$ & $0.952$ & $+0.002$ \\
$0.005$ & $0.90$ & $1.944$ & $1.779$ & $8.5\%$  & $0.936$ & $0.905$ & $+0.005$ \\
$0.005$ & $0.80$ & $2.069$ & $1.867$ & $9.8\%$  & $0.859$ & $0.806$ & $+0.006$ \\
$0.010$ & $0.95$ & $1.961$ & $1.799$ & $8.3\%$  & $0.969$ & $0.954$ & $+0.004$ \\
$0.010$ & $0.90$ & $2.061$ & $1.871$ & $9.2\%$  & $0.936$ & $0.907$ & $+0.007$ \\
$0.010$ & $0.80$ & $2.216$ & $1.979$ & $10.7\%$ & $0.865$ & $0.807$ & $+0.007$ \\
$0.025$ & $0.95$ & $2.146$ & $1.945$ & $9.4\%$  & $0.972$ & $0.953$ & $+0.003$ \\
$0.025$ & $0.90$ & $2.282$ & $2.041$ & $10.6\%$ & $0.939$ & $0.908$ & $+0.008$ \\
$0.025$ & $0.80$ & $2.500$ & $2.189$ & $12.4\%$ & $0.873$ & $0.815$ & $+0.015$ \\
$0.050$ & $0.95$ & $2.354$ & $2.107$ & $10.5\%$ & $0.973$ & $0.954$ & $+0.004$ \\
$0.050$ & $0.90$ & $2.537$ & $2.234$ & $11.9\%$ & $0.944$ & $0.908$ & $+0.008$ \\
$0.050$ & $0.80$ & $2.839$ & $2.437$ & $14.2\%$ & $0.879$ & $0.820$ & $+0.020$ \\
$0.100$ & $0.95$ & $2.677$ & $2.356$ & $12.0\%$ & $0.975$ & $0.957$ & $+0.007$ \\
$0.100$ & $0.90$ & $2.945$ & $2.539$ & $13.8\%$ & $0.948$ & $0.913$ & $+0.013$ \\
$0.100$ & $0.80$ & $3.413$ & $2.846$ & $16.6\%$ & $0.889$ & $0.825$ & $+0.025$ \\
\bottomrule
\end{tabular}
\caption{Extended grid validation, Maharaj boundary. Gaussian DGP,
$\Delta = 0.2$, $r = 1$, $m = 20$, $B = 50{,}000$. MC SE is
below $0.002$ at every cell. The Maharaj grid
extends to $\alpha = 0.001$ because $\lambda_{\mathrm{M}}(\alpha)$
has a closed form for any $\alpha$.}
\label{tab:extended_maharaj}
\end{table}

\begin{table}[ht]
\centering
\small
\begin{tabular}{cccccccc}
\toprule
$\alpha$ & $1-\beta$ & $\knaiveJ$ & $\kstarJ$ & saving &
$\hat\pi(\knaiveJ)$ & $\hat\pi(\kstarJ)$ & gap \\
\midrule
$0.010$ & $0.95$ & $1.750$ & $1.593$ & $9.0\%$  & $0.970$ & $0.952$ & $+0.002$ \\
$0.010$ & $0.90$ & $1.829$ & $1.647$ & $9.9\%$  & $0.937$ & $0.902$ & $+0.002$ \\
$0.010$ & $0.80$ & $1.953$ & $1.732$ & $11.3\%$ & $0.865$ & $0.799$ & $-0.001$ \\
$0.025$ & $0.95$ & $1.901$ & $1.706$ & $10.2\%$ & $0.973$ & $0.952$ & $+0.002$ \\
$0.025$ & $0.90$ & $2.008$ & $1.779$ & $11.4\%$ & $0.939$ & $0.904$ & $+0.004$ \\
$0.025$ & $0.80$ & $2.183$ & $1.897$ & $13.1\%$ & $0.867$ & $0.802$ & $+0.002$ \\
$0.050$ & $0.95$ & $2.072$ & $1.834$ & $11.5\%$ & $0.973$ & $0.952$ & $+0.002$ \\
$0.050$ & $0.90$ & $2.217$ & $1.930$ & $12.9\%$ & $0.943$ & $0.905$ & $+0.005$ \\
$0.050$ & $0.80$ & $2.460$ & $2.092$ & $15.0\%$ & $0.873$ & $0.806$ & $+0.006$ \\
$0.100$ & $0.95$ & $2.343$ & $2.033$ & $13.2\%$ & $0.974$ & $0.952$ & $+0.002$ \\
$0.100$ & $0.90$ & $2.556$ & $2.172$ & $15.0\%$ & $0.945$ & $0.904$ & $+0.004$ \\
$0.100$ & $0.80$ & $2.933$ & $2.419$ & $17.5\%$ & $0.879$ & $0.804$ & $+0.004$ \\
\bottomrule
\end{tabular}
\caption{Extended grid validation, mSPRT boundary with $\sigma_\tau = \delta_{\mathrm{MDE}}$. Gaussian DGP, $\Delta = 0.2$,
$r = 1$, $m = 20$, $B = 50{,}000$. MC SE is below $0.002$
at every cell.}
\label{tab:extended_msprt}
\end{table}

\begin{table}[ht]
\centering
\small
\begin{tabular}{cccccl}
\toprule
$p_C$ & $\frac{\nz}{2} p_C$ & $\hat\pi(\knaive)$ & $\hat\pi(\kstar)$ & gap & \\
\midrule
$0.20$  & $247$ & $0.872$ & $0.795$ & $-0.005$ & \\
$0.05$  & $62$  & $0.860$ & $0.788$ & $-0.012$ & \\
$0.01$  & $12$  & $0.841$ & $0.768$ & $-0.032$ & CLT strained \\
$0.001$ & $1$   & $0.799$ & $0.727$ & $-0.073$ & CLT strained \\
\bottomrule
\end{tabular}
\caption{Low-base-rate Bernoulli validation (WSKR boundary).
$\alpha = 0.05$, $1-\beta = 0.80$, $\Delta = 0.1$, $r = 1$,
$B = 50{,}000$. The column $(\nz/2) \cdot p_C$ is the expected number of
successes per arm at the fixed-sample size ($r = 1$, so each arm
receives $\nz / 2$ observations). Below roughly $10$ expected
successes the CLT approximation degrades and the closed form
undershoots target power.}
\label{tab:lowbase}
\end{table}

\end{document}